\documentclass[conference]{IEEEtran}

\usepackage[T1]{fontenc}
\usepackage[latin9]{inputenc}
\usepackage{setspace,caption}
\usepackage[linesnumbered,ruled,vlined]{algorithm2e}
\usepackage{algpseudocode}
\usepackage{amsmath, amsfonts, amsthm, amssymb, wasysym, txfonts, bm, bbm,graphicx}
\usepackage{epstopdf,multicol} 
\usepackage[font=small,labelfont=bf]{caption}
\usepackage{pdfpages, float}
\usepackage{mathrsfs}
\usepackage{graphicx}

\newtheorem{lem}{Lemma}

\newtheorem{cor}{Corollary}
\newtheorem{rem}{Remark}
\newtheorem{assumption}{Assumption}

\newtheorem{definition}{Definition}

\newtheorem{example}{Example}

\newtheorem{prop}{Proposition}

\begin{document}
%
\title{Algebraic Optimization of Binary Spatially Coupled Measurement Matrices for Interval Passing }

\author{\thanks{This work is supported in part by NSF grants ECCS-1407910
    and ECCS-1711056.}



\IEEEauthorblockN{Salman Habib and J{\"o}rg Kliewer}
\IEEEauthorblockA{
	Helen and John C.~Hartmann Department of Electrical and Computer Engineering \\
	New Jersey Institute of Technology, Newark, NJ 07102\\
	Email: sh383@njit.edu, jkliewer@njit.edu}
}
\maketitle

\begin{abstract} 
We consider binary spatially coupled (SC) low density measurement matrices
for low complexity reconstruction of sparse signals via the interval passing
algorithm (IPA). The IPA is known to fail due to the presence of harmful
sub-structures in the Tanner graph of a binary sparse measurement matrix, so
called termatiko sets. In this work we construct  array-based (AB) SC sparse
measurement matrices via algebraic lifts of graphs, such that the number of
termatiko sets in the Tanner graph is minimized. To this end, we show for
the column-weight-three case that the most critical termatiko sets can be
removed by eliminating all length-12 cycles associated with the Tanner
graph, via algebraic lifting. As a consequence, IPA-based reconstruction
with SC measurement matrices is able to provide an almost error free
reconstruction for significantly denser signal vectors compared to uncoupled
AB LDPC measurement matrices. 
\end{abstract}


\vspace{-2ex}
\section{Introduction}
\vspace{-1ex}
\label{Intro}

Compressed sensing \cite{CT06}, \cite{Don06} is a tool for estimating a sparse signal $\boldsymbol{x}\in \mathbb{R}^n$ of sparsity order $k$ from a compressed version of the signal $\boldsymbol{y}\in \mathbb{R}^m$, where $k\ll n$ and $m\ll n$. The compressed signal can be obtained by taking $m$ random linear projections of the original signal via the operation $\boldsymbol{y}=A\boldsymbol{x}$, where $A\in \mathbb{R}^{m\times n}$ is an $m\times n$ measurement matrix. 
 
A straightforward way for reconstructing the signal is to find a vector $\boldsymbol{\hat{x}}$ with the smallest $l_0$ norm. However, as its complexity is NP-hard, this approach is rendered infeasible for most practical applications \cite{Don06}. A more efficient approach based on linear programming (LP), called Basis Pursuit, has been proposed in \cite{CDS98}, which however is still too complex for applications that require fast reconstruction. To overcome these complexity issues, message passing schemes such as verification decoding and iterative thresholding algorithms have been proposed for reconstructing compressed signals \cite{WY13}, \cite{MalDon10}. An improved messaging passing algorithm known as Approximate Message Passing (AMP) is proposed in \cite{DonMalMon09}, which has an identical sparsity to sampling ratio trade-off as LP, albeit at a much lower computational complexity. 

The interval-passing algorithm (IPA) was first proposed in \cite{CSW10} for both binary and non-negative real measurement matrices. For  measurement matrices derived from parity check matrices of LDPC codes, the IPA is known to fail due to the presence of stopping sets. In particular, in \cite{Ravn12} it is shown that if the  Tanner graph associated with the support of a signal $\boldsymbol{x}$ contains a non-empty stopping set, then the IPA fails to fully recover $\boldsymbol{x}$, but some of the samples inside these sets can be recovered. In \cite{Yauh16} a complete graphical description of harmful substructures causing a recovery failure, the so called termatiko sets, is provided.
In particular, if the Tanner graph associated with the support of $\boldsymbol{x}$ contains a termatiko set, then the IPA completely fails to recover the signal. 

In this work we are mainly interested in the reconstruction performance of
array-based (AB) spatially coupled (SC) measurement matrices, obtained by
coupling regular AB LDPC code-based measurement matrices. Note that AB SC
LDPC codes can be constructed via an edge-spreading process applied to a
base Tanner graph of the LDPC block code (BC), yielding an SC
protograph. 
Recently, general edge-spreading schemes \cite{MR17} have been proposed as
an extension of the widely used cutting vector approach \cite{FZ99} for
constructing SC codes. Additionally, \cite{MR17} considers the design of
generalized cutting vectors with the objective of maximizing the minimum
distance of the corresponding SC protograph, thus also maximizing the size
of the smallest stopping set in the Tanner graph of the code \cite{R15}.
%
%
In \cite{ASCJ17} we have proposed a new algebraic lifting
strategy for constructing AB SC LDPC codes, which outperforms existing
schemes  in terms of reducing critical substructures in the
Tanner graph of the AB SC code.

Also, it is known that AB block measurement matrices are able
to outperform Gaussian measurement matrices under AMP decoding
\cite{LX13}. Further, in \cite{Pfister10} it is shown that SC LDPC
measurement matrices obtained from randomly generated regular LDPC BCs
outperform uncoupled measurement matrices under verification
decoding. However, to the best of our knowledge, the use of binary  AB SC
LDPC code-based measurement matrices under IPA reconstruction has not been
studied so far. In particular, we
  propose to construct binary SC AB measurement matrices 
such that the number of both size-three and size-six termatiko sets in the underlying Tanner graph is
  minimized. As one of our main results we show that for the
  column-weight-three AB case, these termatiko sets can be removed
  efficiently by eliminating length-12 cycles in the Tanner graph.
As a consequence, IPA-based reconstruction in conjunction with binary SC LDPC code based measurement matrices is able to provide a low complexity, almost error free reconstruction for significantly denser signal vectors compared to uncoupled AB LDPC based measurement matrices.



\vspace{-0.5ex}
\section{Preliminaries}
\vspace{-0.5ex}

\subsection{Algebraic lifting}
\label{sec:alg_lft}

Let the Tanner graph associated to a $m\times n$ binary matrix $A$ be
represented by $G=(V\cup C,E)$, where $V=\{v_1,v_2,\ldots,v_n\}$ is a set of
variable nodes (VNs), $C=\{c_1,c_2,\ldots,c_m\}$ is a set of check nodes
(CNs), and $E= \{(v_i,c_j)|v_i\in V, c_j\in C, A(j,i)=1\}$ is the set of
edges connecting $v_i$ to $c_j$, for $i=\{1,\ldots,m\}$ and
$j=\{1,\ldots,n\}$. We also denote the set of neighbors for each node $v_i$
and $c_j$ as $\mathcal{N}(v_i)=\{c_j\in C|(v_i,c_j)\in E\}$ and
$\mathcal{N}(c_j)=\{v_i\in V|(v_i,c_j)\in E\}$, respectively. In general, a
{\em degree $J$ lift} of $G$ is a graph $\hat{G}$ with VN set $\hat{V} =
\{v_{1_1},\ldots, v_{1_J},\ldots, v_{n_1},\ldots, v_{n_J}\}$ of size $nJ$
and CN set $\hat{C} = \{c_{1_1},\ldots, c_{1_J},\ldots, c_{m_1},\ldots,
c_{m_J}\}$ of size $mJ$ and for each $e \in E$, if $e = (v_i,c_j)$ in $G$,
then there are $J$ edges from $\{v_{i_1},\ldots, v_{i_J}\}$ to
$\{c_{j_1},\ldots, c_{j_J}\}$ in $\hat{G}$ in a one-to-one mapping. The
graph $\hat{G}$ can be obtained algebraically by assigning
permutations to each of the edges in $G$ so that if $e=(v_i,c_j)$ is
assigned a permutation $\tau(k)\in\{1,\ldots,J\}$, the corresponding edges in $\hat{G}$ are $(v_{i_k},c_{j_{\tau(k)}})$ for $1\leq k\leq J$.

The protograph approach  to the construction of SC LDPC codes involves the
base Tanner graph with a parity check matrix represented as $H(\gamma,p)$, where $p$ is odd and $p\geq \gamma$. In case of AB codes, this matrix is given as
\[ H(\gamma,p)=
\begin{bmatrix}
I & I & I & \cdots & I \\
I & \sigma & \sigma^{2} & \cdots & \sigma^{p-1} \\
\vdots & \vdots & \vdots & \cdots & \vdots \\
I & \sigma^{\gamma-1} & \sigma^{2(\gamma-1)} & \cdots & \sigma^{(\gamma-1)(p-1)}
\end{bmatrix},
\]
where $I$ and $\sigma^z$ for $z\in \{1,\ldots,2(p-1)\}$ are identity and
permutation matrices, resp., of dimension $p \times p$. This matrix
can also be considered as a 2-D array of submatrices where each row (column)
of matrices denotes a row (column) \textit{group} with $p$ column groups and
$\gamma$ row groups in total. A SC protograph is then obtained from
$H(\gamma,p)$ via \emph{edge-spreading}. The idea is to split $H(\gamma, p)$
into a sum of $m+1$ matrices of the same dimension as $H(\gamma,
p)=H_0+H_1+,\ldots,H_m$, where $m$ represents the memory of the code. These
matrices are arranged as
\[ H(\gamma, p, L)=
\begin{bmatrix}
H_{0} & & & & \\
H_{1} & H_{0} & & &  \\
& \ddots & \ddots & & \\
& & H_{m} & \cdots & H_{0}  \\
& & & \ddots & \vdots \\
& & & & H_{m}
\end{bmatrix}
\]
to form the parity-check matrix of a terminated SC protograph $H(\gamma, p,
L)\in \mathbb{F}_2^{\gamma (L+1)p\, \times\, Lp^2}$, where $L$ is the
coupling length and $\mathbb{F}_2$ is the binary field.
The final SC LDPC measurement matrix $H(\gamma,p,L,J) \in
\mathbb{F}_2^{\gamma (L+1)Jp\, \times\, LJp^2}$ is then obtained by a {\em terminal lift} of $H(\gamma, p, L)$, where $J$ is the terminal lifting parameter.

For algebraic lifting, let $\tau_L^\kappa$ be a $L\times L$ permutation matrix obtained by left shifting the identity permutation by an amount of $\kappa$, where $0\leq \kappa \leq m$. Then, the SC protograph corresponding to $H(\gamma,p,L)$ can alternatively be constructed by lifting each of the edges of the base protograph by a $\tau_L^\kappa$ matrix; this is equivalent to replacing the non-zero entires of $H(\gamma,p)$ by
$\tau_L^\kappa$, and the zero entries by all-zero matrices of the same size,
respectively. In the same way, the check matrix $H(\gamma,p,L,J)$ of the
final  SC LDPC code can be obtained by lifting each of the edges of the SC
protograph by any $J\times J$ permutation~matrix~\cite{ASCJ17}. 

\subsection{Compressed sensing and the IPA}

Let $\boldsymbol{x}\in \mathbb{R}^n$ be an $n$ dimensional $k$-sparse signal (which means it has at most $k$ nonzero entries). We consider the recovery of $\boldsymbol{x}$ from measurements $\boldsymbol{y} = A\boldsymbol{x} \in \mathbb{R}^m$, where $m \ll n$ and $k \ll n$, and $A$ is the $m\times n$ measurement matrix. The IPA denotes an iterative algorithm $\text{IPA}(y,A)$ to reconstruct a nonnegative real signal $\boldsymbol{x}\in \mathbb{R}^n$ from a measurement vector $\boldsymbol{y}$. It has been stated in \cite{Yauh16} that the IPA reconstruction performance is independent of whether binary or non-negative measurement matrices $A$ and signal vectors $\boldsymbol{x}$ are used. Therefore, without loss of generality we consider $A\in\mathbb{F}_2^{ \times n}$ and $\boldsymbol{x}\in\mathbb{F}_2^{n}$. 
We also denote the elements of $\boldsymbol{x}$ and $\boldsymbol{y}$ as $\boldsymbol{x}=[x(v_1),\ldots,x(v_n)]^T$ and $\boldsymbol{y}=[y(c_1),\ldots,y(c_m)]^T$, respectively. Recovery takes place by iteratively exchanging messages on the Tanner graph of $A$, where the measurement nodes will be denoted as VNs and the function nodes as CNs in the following.

\vspace{-0.2ex}
\subsection{Stopping sets in AB measurement matrices}
\label{sec:stopset}
Stopping sets are harmful structures in the Tanner graph of $A$ that can
cause the IPA to fail. In this work, we first analyze the structure of
minimum stopping sets in AB SC LDPC  measurement matrices which are the  most harmful to the IPA decoder. 

\vspace{-0.5ex}
\begin{definition}[\cite{Proietti02}]
\label{def:stpset}
A stopping set $S(M)=\{v_1,\ldots,v_M\} \subset V$  is a non-empty subset of
the set of $M$ variable nodes $V$ such that all neighbors of $S(M)$ are connected to it at least twice.
\end{definition}

In the following, we focus on AB parity check matrix with a column weight of $\gamma=3$ for the sake of simplicity\footnote{Note that in the following "AB codes" refers to AB codes with $\gamma=3$.}. For $\gamma=3$, each VN has $3$ neighbors, so there must be $3M$ edges connected to $S(M)$. The Tanner graph of an AB code also consists of \emph{cycles} with the following structural properties.

\begin{rem}
\label{rem:cyc}
	For $\gamma=3$, a cycle of length $\ell$ in an AB code consists of $\ell$ edges that are connected to $\ell/2$ CNs (resp. VNs) of degree $2$ with respect to the VNs (resp., CNs) of the cycle. By the pattern consistency condition \cite{DAnantharam10} each CN associated to the cycle is connected to a distinct pair of VNs of the cycle.  
\end{rem}

We now address  the structure of some \emph{small} stopping sets of size $\leq 12$. Let $\mathcal{N}(S(M))$ be the set of all neighboring CNs of $S(M)$, and let $e(S(M))$ be the set of all the edges connecting $S(M)$ to $\mathcal{N}(S(M))$. 

\begin{rem}
\label{rem:min_stop}
	For $\gamma=3$, the minimum stopping set $S(6)$ in an AB code consists of $6$ degree $3$ VNs that are connected via $18$ edges to $9$ CNs of degree $2$ with respect to the VNs in $S(6)$. 
\end{rem}

Since there are no 4-cycles in AB codes \cite{DAnantharam10}, a pair of
neighbors of $S(6)$ cannot be connected to the same pair of VNs in
$S(6)$. In $e(S(6))$, nine CNs are connected to nine VN pairs via $18$
edges. There are ${6\choose 2}=15$ possible pairs of VNs in $S(6)$. However,
we only consider nine pairs such that each VN of $S(6)$ appears in exactly
three pairs out of those nine (because the VN degree is $3$). This also implies that there exists two pairs out of those nine that have a common VN, and that is true for all VNs in $S(6)$. Thus, we obtain the following lemma.

\begin{lem}
\label{lem:stop_set}
	 For $\gamma=3$, the minimum stopping set $S(6)$ in an AB code
         consists of six VNs of degree $3$ that are connected via $18$ edges
         to nine CNs of degree $2$ with respect to the VNs in $S(6)$. Here,
         a pair of neighboring CNs of $S(6)$, denoted by $c$ and $c'$,
         respectively, that have a common neighbor $v_k$, must be connected
         to three VNs $\{v_i,v_k,v_q\}\in S(6)$ via four  edges $(c,v_i)$, $(c,v_k)$, $(c',v_k)$ and $(c',v_q)$, where $i\neq q\neq k$. 
\end{lem}

\vspace{-1ex}

\vspace{-1ex}
\begin{cor}
\label{cor:type3CN}
	 There are two sets of VNs $V'\subset S(6)$ and $\hat{V}= S(6)\setminus V'$, where $|V'|=|\hat{V}|=3$, that are connected to all CNs in $\mathcal{N}(S(6))$. 
\end{cor}

\vspace{-2ex}

\vspace{-0.5ex}
\section{Termatiko Sets in AB Measurement Matrices}
\vspace{-0.5ex}
\subsection{Preliminaries}
In \cite{Yauh16} it is  shown that stopping sets may not cause a total failure of the IPA. Under some conditions, some of the non-zero values of the signal can be recovered even if the VNs in the Tanner graph of the measurement matrix corresponding to the non-zero values are associated with a stopping set. However, there are sets of VNs inside a stopping set, termed termatiko sets, that cause a total failure of the IPA if the support of $\boldsymbol{x}$, $\text{supp}(\boldsymbol{x})=\{v\in V: x(v)\in \boldsymbol{x}, x(v)\neq 0\}$, is a termatiko set. 

\vspace{-0ex}
\begin{definition}[\cite{Yauh16}]
\label{def:tset}
A subset $T_{w,M}\subseteq S(M)$ is a termatiko set of size $w\leq M$ if and only if the function $\text{IPA}(A\boldsymbol{x}_{T_{w,M}},A)$ returns $\hat{\boldsymbol{x}}=\boldsymbol{0}$, where $\boldsymbol{x}_{T_{w,M}}$ is a binary vector with  $\textrm{supp}(\boldsymbol{x}_{T_{w,M}})=T_{w,M}$. 
\end{definition}
\vspace{-0.5ex}

We denote by $N$ the set of CNs connected to $T_{w,M}$. Moreover, we denote by $\hat{S}=\{v\in V\setminus T_{w,M}: \mathcal{N}_N(v)=\mathcal{N}(v)\}$ the set of remaining VNs outside $T_{w,M}$ connected only to $N$, where $\mathcal{N}_N(v)$ is the set of neighbors of $v$ in $N$. $T_{w,M}$ exists only if for each $c\in N$ one of the following conditions is true \cite{Yauh16}:


\begin{enumerate}
\item[(i)] A CN $c \in N$ is connected to $\hat{S}$.
\item[(ii)] If $c \in N$ is not connected to $\hat{S}$, then it must have at
  least two neighbors belonging to set $T_{w,M}$ satisfying the following
  constraint: all CNs $c'\! \in\! N$ connected to these neighbors must  have
  at least two neighbors in $T_{w,M}$.

\end{enumerate}

\subsection{Minimum termatiko sets}
In the following we  analyze the structure of minimum termatiko sets, residing in minimum stopping sets, which have the smallest possible value of $w>0$. For AB measurement matrices with $\gamma=3$, this minimum termatiko set of this type is denoted as $T_{3,6}$. 
\vspace{-1ex}

\begin{prop}[see also \cite{YR18}]
\label{prop:tset3}
A set of three VNs in $S(6)$ constitutes a $T_{3,6}$ set if it is connected to all nine CNs in $\mathcal{N}(S(6))$. Also, a $S(6)$ stopping set consists of two $T_{3,6}$ sets.
\end{prop}
\vspace{-1ex}

\begin{proof}
As in \cite{Yauh16} we assume without loss of generality that $V=T_{w,M}\cup
\hat{S}$. Recall from Corollary \ref{cor:type3CN} that there exists two sets
of VNs $V'\subset S(6)$ and $\hat{V}= S(6)\setminus V'$, where
$|V'|=|\hat{V}|=3$, and each of them are connected to all nine  CNs in $\mathcal{N}(S(6))$. Thus, according to Condition (i) above we obtain that $T_{3,6}=V',\hat{S}=\hat{V}$, and $T_{3,6}=\hat{V},\hat{S}=V'$, respectively, and $N=\mathcal{N}(S(6))$.

On the other hand, assume now that a set of three VNs $\tilde{V}\subset S(6)$ is not equal to $V'$ or $\hat{V}$. Then, according to the structure of $e(S(6))$, the total number of CNs connected to $\tilde{V}$ is less than nine. If we assume for a moment that $\tilde{V}= T_{3,6}$, then this would imply that $|N|<|\mathcal{N}(S(6))|=9$; in other words $N\neq \mathcal{N}(S(6))$. Then, due to the properties of $e(S(6))$, there would be less than nine CNs in the set $\mathcal{N}(S(6))\setminus N$ that has neighbors in the set $\tilde{S}=S(6)\setminus \tilde{V}=S(6)\setminus T_{3,6}$. However, this also implies that $\tilde{S}\neq \hat{S}$. Consequently, $\tilde{V}\neq T_{3,6}$.
\end{proof}

Fig. \ref{fig:not_tset} shows that a set of VNs $\{v_2,v_3,v_5\}$ is not connected to all the neighbors of $S(6)$, hence it cannot form a $T_{3,6}$ set, since if it did, it would imply that
 $\hat{S}=\{v_1,v_4,v_6\}$. This contradicts the definition of $\hat{S}$ as the set $\{v_1,v_4,v_6\}$ is not connected to all green CNs (which are neighbors of the candidate termatiko set $\{v_2,v_3,v_5\}$ in this example). 


\begin{figure}[H]
\centering
\includegraphics[scale=1.1]{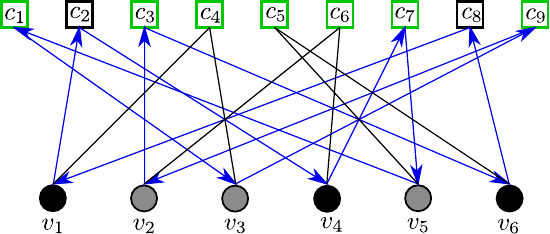}
\caption{Example of a case where a set of VNs $\{v_2,v_3,v_5\}$ cannot form a termatiko set in $S(6)=\{v_1,v_2,v_3,v_4,v_5,v_6\}$. The underlying 12-cycle is shown in blue. The VNs $\{v_1,v_2,v_5\}$ are connected to all neighbors of $S(6)$ and therefore  represent a $T_{3,6}$ termatiko set with $\hat{S}=\{v_3,v_4,v_6\}$.}
\vspace{-2.5ex}
\label{fig:not_tset}
\end{figure}


\vspace{-1ex}
\begin{rem}
\label{rem:T_6_6}
	From the proof of Proposition 1 we have seen that a $T_{3,6}$ set can exist in $S(6)$ in two possible configurations: $T_{3,6}=V',\hat{S}=\tilde{V}$, and $T_{3,6}=\tilde{V},\hat{S}=V'$. In other words, both $V'$ and $\tilde{V}$ are termatiko sets. The fact that $V'\cup \hat{V}=S(6)$  satisfies Condition (ii) above with the set $\hat{S}$ being the empty set implies that $S(6)$ is also a $T_{6,6}$ termatiko set. 
\end{rem}

\begin{lem}
\label{lem:stp_12cycle}
	 A $T_{6,6}$ set contains at least two 12-cycles. 
\end{lem}

\begin{proof}
Consider the Tanner graph of an AB code with $m$ CNs and $n$ VNs. Let $V_1:\{v_{i_1},v_{i_2},v_{i_3},v_{i_4},v_{i_5},v_{i_6}\} \subset V$ and $\hat{C}:\{c_{j_1},c_{j_2},\ldots,c_{j_9}\}\subset C$, where $j_k\in\{1,2,\ldots,m\}$, $k\in\{1,2,\ldots,9\}$ and $i_{\ell}\in\{1,2,\ldots,n\}$, $\ell \in\{1,2,\ldots,6\}$. We split $\hat{C}$ into three subsets, $C_1$, $C_2$ and $C_3$, respectively, where $C_1:\{c_{j_1},c_{j_2},c_{j_3}\}$, $C_2:\{c_{j_4},c_{j_5},c_{j_6}\}$, and $C_3:\{c_{j_7},c_{j_8},c_{j_9}\}$. We now establish a condition under which the VNs in $V_1$ connected to CNs in set $C_1\cup C_2$ are associated to a 12-cycle. The six edges connecting $V_1$ to $C_1$ and $C_2$, respectively, are denoted as $e_1$ and $e_2$, respectively. Next, we establish a condition under which the VNs in $V_1$ connected to CNs in set $C_1\cup C_3$ are associated to another 12-cycle, where the six edges connecting $V_1$ to $C_3$ are denoted as $e_3$. Finally, we show that under these conditions $e(S(6))=e(T_{6,6})=e_1\cup e_2\cup e_3$ by invoking Lemma \ref{lem:stop_set}. Further details are omitted in the interest of space.
\end{proof}

\subsection{Other termatiko sets associated to 12-cycles}

\begin{rem}
\label{rem4}
In the same way as above we can show that an $S(8)$ stopping set contains a 12-cycle
whose VNs form a $T_{6,8}$, and that an $S(12)$ stopping set contains a 12-cycle
whose VNs form a $T_{6,12}$, respectively. Details are omitted due to
space constraints. Since $|\mathcal{N}(T_{6,8})|\neq |\mathcal{N}(S(6)|$ we can
conclude that $T_{6,8}\neq S(6)$. Likewise, $T_{6,12}\neq S(6)$.
\end{rem}

\subsection{Eliminating small termatiko sets via algebraic lifting}
 In the algebraic lifting process described in Section \ref{sec:alg_lft}, a $\ell$-cycle can be broken by the lift if we ensure that the net permutation, which is the product of the oriented edge labels, assigned to its edges is not identical to the identity permutation. Let the assignments to the edges of a $\ell$-cycle be $\tau_L^{\kappa_1},\ldots,\tau_L^{\kappa_\ell}$, where $\tau_L^\kappa$ is a permutation matrix as discussed in Section \ref{sec:alg_lft}. Without loss of generality, then, the net permutation of the cycle is given by $\tau_L^{\sum_{i=1}^{\ell}(-1)^{i+1}k_i}$. This becomes the identity permutation only when 
 
\begin{equation}
\label{eq:permutation}
	\sum_{i=1}^{\ell}(-1)^{i+1}k_i=0,
\end{equation} 
where $0\leq k_i\leq m$. For example, for $\ell=12$, a 12-cycle will be
eliminated by the algebraic lifting process if 
(\ref{eq:permutation}) is non-zero. 



\vspace{-1ex}
\section{Optimization of AB SC Measurement Matrices}

In our previous work \cite{ASCJ17} we have shown that all harmful $(3,3)$
absorbing sets can be removed from an AB SC protograph by eliminating all
6-cycles due to the fact that each $(3,3)$ absorbing set contains a
6-cycle. In the same fashion we can see from  Lemma
  \ref{lem:stp_12cycle} and Remark~\ref{rem4}, that if we remove all
  12-cycles via a properly chosen algebraic lifting, we can eliminate all
  $T_{6,\{6,8,12\}}$ termatiko sets. Since $T_{6,6}$ sets include two
  $T_{3,6}$ sets, by this method we can also remove all $T_{3,6}$ termatiko
  sets.  In the following, we focus on two lifting schemes for constructing
  the SC protograph, namely cutting vector based \cite{AREKD16} and algebraic lifting schemes~\cite{ASCJ17}.

\vspace{-1ex}
\subsection{Enumeration of termatiko sets of size 6}
\label{sec:enum}

  Let $C(12)\subset V$, $|C(12)|=6$, represent the six  VNs of a 12-cycle in
  $G$. From Lemma \ref{lem:stp_12cycle} and Remark~\ref{rem4} it is evident
  that $T_{6,\{6,8,12\}}$ sets  are in fact $C(12)$ sets.
 %
 Let $\mathscr{C}_{12}$ denote the set of all unique $C(12)$ sets, i.e., all
 12-cycles with a different set of VNs. In order to find the VN index $i$
 associated to an edge $(v_i,c_j)$ of a 12-cycle in $G$, we employ a cycle
 detection algorithm, such as the improved message passing algorithm
 proposed in \cite{LiLin15}. Such an algorithm has polynomial complexity and
 for AB codes, the complexity can be further reduced by factor $p$. We then obtain the set $\mathscr{C}_{12}$ by
 employing an efficient (binary) search algorithm to detect duplicate cycles
 associated with the same set of VNs.  This search algorithm has a
 complexity of $O(\log \mu)$, where $\mu$ is the number of all detected
 (non-unique) 12 cycles, which potentially can be very large. Algorithm
 \ref{alg:t_enum} proposes a simple enumeration algorithm for all
 $T_{6,M}$ sets,  $\mu_{T_{6,M}}$ with $M\in\{6,7,\dots,L p^2\}$, associated
 with a 12-cycle. Note that all $C(12)$ sets in $G$ are not necessarily
 associated to a termatiko set of size $6$. In order to determine whether or
 not a $C(12)$ set is a $T_{6,M}$ set, we adopt the following rule in
 Algorithm  \ref{alg:t_enum}: If and only if the IPA outputs a vector $\boldsymbol{\hat{x}}=\boldsymbol{0}$ corresponding to an input data vector $\boldsymbol{x}$ with support $C(12)$, represented as $\boldsymbol{x}_{C(12)},$ then $C(12)=T_{6,M}$. Note that $\mu_{T_{6,6}}+\mu_{T_{6,8}}+\mu_{T_{6,12}}\leq
  \mu_{T_{6,M}} \forall M\in\{6,7,\dots,L p^2\}$, and equality holds if
  $T_{6,\{6,8,12\}}$ are the only size $6$ termatiko sets
  associated to 12-cycles.

\setlength{\textfloatsep}{5pt}
\begin{algorithm}[h]
\caption{Enumeration of all $T_{6,M}$ sets with
    $M\in\{6,7,\dots,L p^2\}$ in an AB measurement matrix $A$ ($\gamma=3$)}
\SetAlgoLined
\DontPrintSemicolon
\SetKwInOut{Input}{Input}
\SetKwInOut{Output}{Output}
\Input{$\mathscr{C}_{12},A$}
\Output{$\mu_{T_{6,M}}$}
\label{alg:t_enum}

Initialization: 
	$\mu_{T_{6,M}}=0$\;
\ForEach{$C(12)\in \mathscr{C}_{12}$}
{
	Fix a binary $\boldsymbol{x}_{C(12)} \text{ with } \text{supp}(\boldsymbol{x}_{C(12)})=C(12)$\;
	Compute $\boldsymbol{y}_{C(12)}=A\boldsymbol{x}_{C(12)}^T$ \;
	Run $\text{IPA}(\boldsymbol{y}_{C(12)},A)$\;
	\If{$\boldsymbol{\hat{x}=0}$}{	
			$\mu_{T_{6,M}}=\mu_{T_{6,M}}+1$ 
	}
}
\end{algorithm}


\mbox{}

\vspace{-4ex}
\subsection{Optimization of the SC protograph}
\label{sec:optim}

Let us define the permutation indicator matrix $B_1\in \{0,1\}^{\gamma \times p}$, where a $1$ (resp., $0$) in position $(i,j)$ of this matrix indicates that all the non-zero elements of block $(i,j)$ of $H(\gamma,p)$ will be lifted by $\tau_L^\kappa$ (resp., $I$), for $\kappa\in\{1,2,\ldots,m\}$, resulting in the $H(\gamma,p,L)$ SC protograph matrix. The process of obtaining optimized SC protographs by using both cutting vector and algebraic lifting approaches is described as follows:

(i) We first choose an $H(3,p)$ AB block matrix.

(ii) For the cutting vector approach based on the $H(3,p)$
  AB block matrix, we construct SC protograph matrices by choosing a cutting
  vector $\boldsymbol{\xi}^*$ from \cite[Table III]{R15} that provides a
  maximal minimum distance of 8 for the AB SC  protograph. For such a code the minimum distance is
  equivalent to the stopping distance \cite{R15}, and therefore it follows
  from Proposition \ref{prop:tset3}
 and Remark \ref{rem:T_6_6} that the AB SC protograph does not contain
  any $T_{3,6}$ and  $T_{6,6}$ termatiko sets.

(iii) In case of algebraic lifting, we minimize the number of 12-cycles in the Tanner graph of the SC protograph obtained from $H(3,p)$. We numerically  optimize the $B_1$ permutation matrix by using the  approach in \cite{ASCJ17}, and the cycle counting algorithm of \cite{LiLin15} is utilized to count the number of  12-cycles in each optimization step. This leads to an optimized SC protograph matrix $H(3,p,L)$ that contains a smaller number of $T_{6,M}$ sets compared to the non-optimized protograph.

(iv) Finally, for both SC protographs discussed previously, we apply a
degree $J$ lift to $H(3,p,\boldsymbol{\xi}^*,L)$ and $H(3,p,L)$, resp., and
obtain the corresponding optimized AB SC measurement matrix $A$, whose Tanner graph is used for reconstruction by the IPA. 


\vspace{-1ex}
\begin{prop}
\label{prop:2}
Let $\hat{G}$ be a Tanner graph obtained by applying a degree $J$ lift to
the Tanner graph $G$. Let $\mu_{C(12)}$ (resp. $\hat{\mu}_{C(12)}$)
represent the total number of 12-cycles in the graph $G$ (resp.,
$\hat{G}$). Also, let $\mu_{T_{3,6}}$, $\mu_{T_{6,M}}$
(resp. $\hat\mu_{T_{3,6}}$, $\hat{\mu}_{T_{6,M}}$) represent the total
number of $T_{3,6}$, $T_{6,M}$ sets in the graph $G$ (resp., $\hat{G}$). We
then have $\hat{\mu}_{C(12)}\leq J\mu_{C(12)}$ and $\hat{\mu}_{T_{3,6}}\leq
J\mu_{T_{3,6}}$, $\hat{\mu}_{T_{6,M}}\leq J\mu_{T_{6,M}}$.
\end{prop}

\vspace{-1ex}
\noindent The proof is a simple consequence of the properties of graph lifting.

\vspace{-2ex}
\section{Simulation Results}
\label{sec:results}
We now provide results for the IPA reconstruction performance for different constructions of measurement matrices via Monte Carlo simulations. 

\begin{itemize}
\item $A_1$ is obtained as a block diagonal matrix where each block is
  obtained from a $H(3,7)$
  AB base matrix of size $3p\times p^2$ and then individually uplifted by factor $J$.
\item $A_2$ represents a non AB SC LDPC matrix obtained by coupling $L$
  copies of a $(3,7)$ \emph{random} regular LDPC matrix of size $3p\times p^2$, uplifted by a factor $J$.
\item $A_3$ represents a $H(3,7,\boldsymbol{\xi}^*,L,J)$ matrix obtained by
  applying a degree $J$ lift to the protograph of the
  $H(3,7,\boldsymbol{\xi}^*,L)$ SC protograph matrix from  a cutting vector approach. 
\item $A_4$ represents a $H(3,7,L,J)$ matrix obtained by applying a degree $J$ lift to the protograph of the  optimized $H(3,7,L)$ SC protograph matrix based on algebraic lifting. 
\item $A_5$ represents a Gaussian matrix with same dimension as $A_4$ whose elements are $\mathcal{N}(0,\sigma^2)$ Gaussian random variables. Without loss of generality, $\sigma^2=1$.
\end{itemize}

The matrices $A_1$ to $A_4$ have the same constraint length of $Jp^2$, and all matrices have dimension $3(L+1)Jp\times LJp^2$. As parameters we select $\gamma=3$, $p=7$, $m=1$, $J=5$, $L=10$, which leads to a blocklength of $n=2450$ for all matrices.
%
For these parameters Table~\ref{tab:tab1}  shows the total number of
12-cycles and $T_{6,M}$ sets, $M\in\{6,7,\dots Lp^2\}$, for the corresponding
protograph matrices of $A_1$, $A_3$ and $A_4$ \footnote{The enumeration
  result for $A_2$ has been excluded as termatiko set enumeration for non-AB
  matrices is beyond the scope of this paper.}. We observe
  that spatial coupling is able to provide a
  significant reduction of both $T_{3,6}$ and $T_{6,\{6,8,12\}}$ termatiko
  sets. Also,  Table~\ref{tab:tab1} verifies that for the cutting vector
  approach with $\boldsymbol{\xi}^*$ all  $T_{3,6}$ sets are eliminated. We
  also see that by optimizing the AB SC measurement matrix via algebraic
  lifting, $T_{6,M}$ sets can be completely removed from the protograph,
  which also implies the elimination of $T_{3,6}$ and $T_{6,\{6,8,12\}}$ sets. By invoking Proposition~\ref{prop:2}, these results also
  hold for the terminally lifted Tanner graph of $A_4$. 

\vspace{0cm}
\begin{table}[h]
\vspace{-1ex}
\centering
\resizebox{\columnwidth}{!}{
\begin{tabular}{|c|c|c|c|c|}
\hline
\textbf{Number of}&\textbf{ protograph of} \boldsymbol{$A_1$}&\textbf{protograph of} \boldsymbol{$A_3$}&\textbf{protograph of} \boldsymbol{$A_4$} \\
\hline\hline
12-cycles&$2409050$&$661311$&$227150$ \\
\hline
$T_{3,6}$ sets&$4900$&$0$&$0$ \\
\hline
$T_{6,M}$ sets&$9800$&$63$&$0$ \\
\hline
\end{tabular}
}
\caption{\label{tab:tab1} Total number of 12-cycles and $T_{6,M}$ sets,
  $M\in\{6,7,\dots,Lp^2\}$,  in
  the corresponding protograph matrices for $A_1$, $A_3$, $A_4$ with the
  parameters 
  $m=1$, $p=7$, $L=10$. }
\end{table}


\vspace{-2ex}

Fig.~\ref{fig:res} displays the IPA reconstruction performance of matrices $A_1$ to $A_4$, and the LP reconstruction performances of $A_4$ and $A_5$. For the IPA, the probability of reconstruction is defined as $\Pr(\boldsymbol{\hat{x}}=\boldsymbol{x})$. For the LP, the  probability of reconstruction is given as $\Pr(\max_{i\in\{1,2,\ldots,n\}}|\hat{x}_i-x_i|\leq 10^{-3})$. All data points on the performance curve are averaged over $1000$ realizations of the binary vector $\boldsymbol{x}$.


From Fig.~\ref{fig:res} we can observe a behavior similar to the results
shown in Table \ref{tab:tab1}, i.e., that spatially coupling leads to a
significant increase in IPA reconstruction performance: for the same
probability of reconstruction the density of the signal can be much higher. 
We also observe that LP based reconstruction for $A_4$ outperforms IPA decoding, albeit at a significantly higher reconstruction complexity. Whereas the IPA has a complexity of only $O(n(\log (n/k))^2 \log(k))$ \cite{CSW10}, LP-based reconstruction has a complexity which is polynomial in time. Therefore, IPA based reconstruction with algebraically lifted  SC measurement matrices serves as a good compromise  between complexity and performance, in particular for larger block lengths.

\begin{figure}
\centering
\vspace{-2ex}
\includegraphics[scale=0.33]{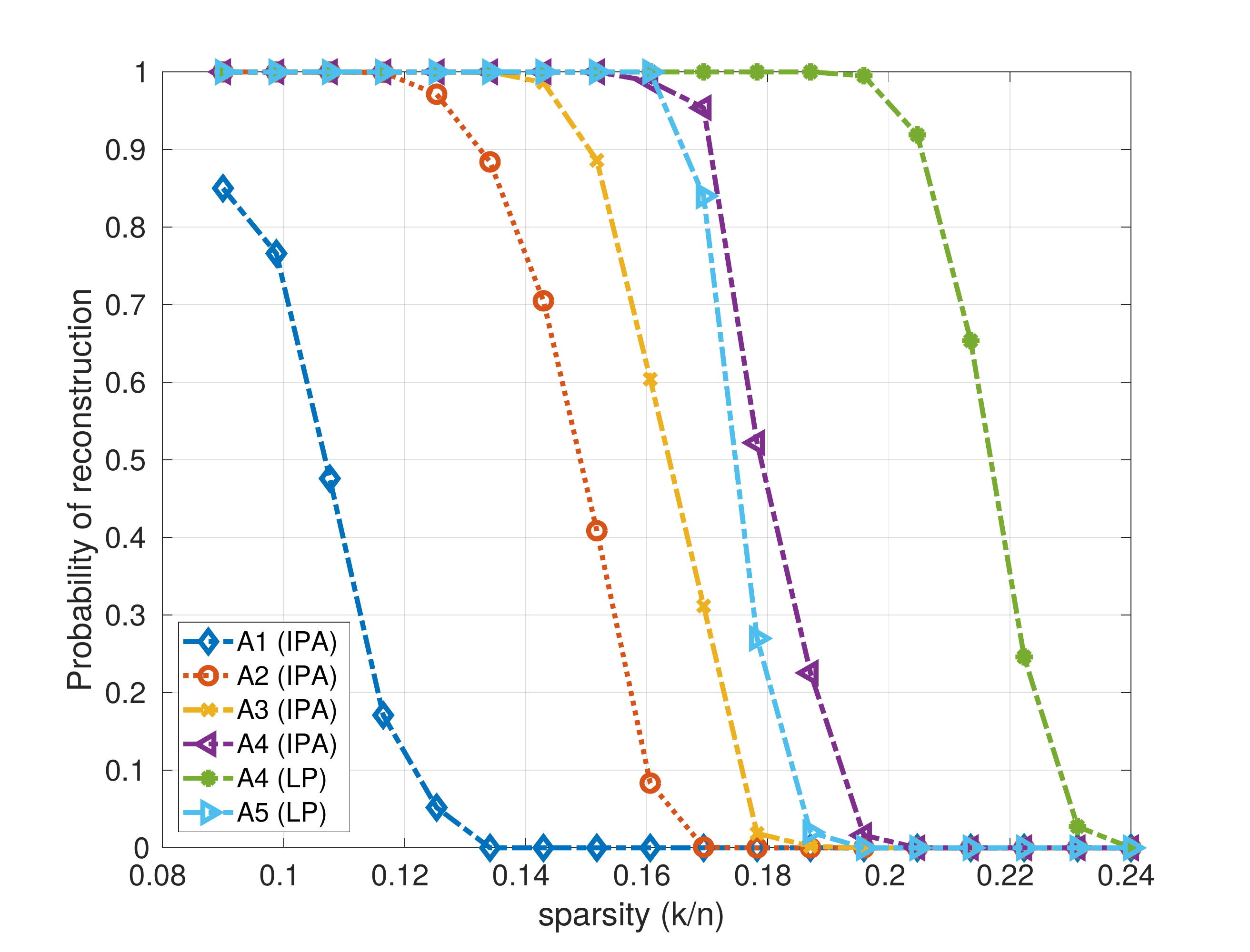}
\vspace{-3ex}
\caption{Reconstruction probability versus sparsity of the data vector
  $\boldsymbol{x}$ for  sparse measurement matrices $A_1$ to $A_4$ with $J=5$,
  $L=10$, $p=7$, $m=1$, and for a Gaussian matrix $A_5$. All matrices have
  dimension $m\times n$ with $m=1155$ and $n=2450$.}
\vspace{-1ex}
\label{fig:res}
\end{figure}

\vspace{-1ex}

\bibliographystyle{IEEEtran}
\bibliography{Bib}

\begin{thebibliography}{10}
\providecommand{\url}[1]{#1}
\csname url@samestyle\endcsname
\providecommand{\newblock}{\relax}
\providecommand{\bibinfo}[2]{#2}
\providecommand{\BIBentrySTDinterwordspacing}{\spaceskip=0pt\relax}
\providecommand{\BIBentryALTinterwordstretchfactor}{4}
\providecommand{\BIBentryALTinterwordspacing}{\spaceskip=\fontdimen2\font plus
\BIBentryALTinterwordstretchfactor\fontdimen3\font minus
  \fontdimen4\font\relax}
\providecommand{\BIBforeignlanguage}[2]{{%
\expandafter\ifx\csname l@#1\endcsname\relax
\typeout{** WARNING: IEEEtran.bst: No hyphenation pattern has been}%
\typeout{** loaded for the language `#1'. Using the pattern for}%
\typeout{** the default language instead.}%
\else
\language=\csname l@#1\endcsname
\fi
#2}}
\providecommand{\BIBdecl}{\relax}
\BIBdecl

\bibitem{CT06}
E.~J. Candes and T.~Tao, ``Near-optimal signal recovery from random
  projections: {U}niversal encoding strategies?'' \emph{IEEE Trans. on Inf.
  Theory}, vol.~52, no.~12, pp. 5406--5425, Dec. 2006.

\bibitem{Don06}
D.~L. Donoho, ``Compressed sensing,'' \emph{IEEE Trans. on Inf. Theory},
  vol.~52, no.~4, pp. 1289--1306, Apr 2006.

\bibitem{CDS98}
S.~S. Chen, D.~L. Donoho, and M.~A. Saunders, ``Atomic decomposition by basis
  pursuit,'' \emph{SIAM J. Comput}, vol.~20, no.~1, pp. 33--61, Aug. 1998.

\bibitem{WY13}
X.~Wu and Z.~Yang, ``Verification-based interval-passing algorithm for
  compressed sensing,'' \emph{IEEE Signal Process. Lett.}, vol.~20, no.~10, pp.
  933--936, Oct. 2013.

\bibitem{MalDon10}
A.~Maleki and D.~L. Donoho, ``Optimally tuned iterative reconstruction
  algorithms for compressed sensing,'' \emph{IEEE J. Sel. Topics Signal
  Process.}, vol.~4, no.~2, pp. 330--341, April 2010.

\bibitem{DonMalMon09}
D.~L. Donoho, A.~Maleki, and A.~Montanari, ``Message passing algorithms for
  compressed sensing: I. motivation and construction,'' in \emph{Proc.~IEEE
  Info. Theory Workshop (ITW), Cairo}, July 2010.

\bibitem{CSW10}
V.~Chandar, D.~Shah, and G.~W. Wornell, ``A simple message-passing algorithm
  for compressed sensing,'' in \emph{Proc. IEEE Int. Symp. Inf. Theory (ISIT)},
  Jun. 2010, pp. 1968--1972.

\bibitem{Ravn12}
V.~Ravanmehr, L.~Danjean, B.~Vasic, and D.~Declercq, ``Interval-passing
  algorithm for non-negative measurement matrices: Performance and
  reconstruction analysis,'' \emph{IEEE J. Sel. Topics Circuits and Systems},
  vol.~2, no.~3, pp. 424--432, Sept. 2012.

\bibitem{Yauh16}
Y.~Yakimenka and E.~Rosnes, ``On failing sets of the interval-passing algorithm
  for compressed sensing,'' in \emph{Proc. 54th Allerton Conf. on
  Communication, Control and Computing}, Sept. 2016, pp. 306--311.

\bibitem{MR17}
D.~Mitchell and E.~Rosnes, ``Edge spreading design of high rate array-based
  {SC-{LDPC}} codes,'' in \emph{Proc. IEEE Int'l Symp. on Info. Theory (ISIT)},
  July 2017.

\bibitem{FZ99}
A.~Jimenez~Felstrom and K.~Zigangirov, ``Time-varying periodic convolutional
  codes with low-density parity-check matrix,'' \emph{IEEE Transactions on
  Information Theory}, vol.~45, no.~6, pp. 2181--2191, Sep 1999.

\bibitem{R15}
E.~Rosnes, ``On the minimum distance of array-based spatially-coupled
  low-density parity-check codes,'' in \emph{Proc. IEEE Int. Symp. Inf. Theory
  (ISIT)}, June 2015.

\bibitem{ASCJ17}
A.~Beemer, S.~Habib, C.~Kelley, and J.~Kliewer, ``A generalized algebraic
  approach to optimizing {SC-{LDPC}} codes,'' in \emph{Proc. 55th Allerton
  Conf. on Communication, Control, and Computing}, Oct. 2017, pp. 1--6.

\bibitem{LX13}
X.~Liu and S.~Xia, ``Constructions of quasi-cyclic measurement matrices based
  on array codes,'' in \emph{Proc. IEEE Int. Symp. Inf. Theory (ISIT)}, Jun.
  2013, pp. 479--483.

\bibitem{Pfister10}
S.~Kudekar and H.~D. Pfister, ``The effect of spatial coupling on compressive
  sensing,'' in \emph{Proc. 48th Allerton Conf. on Communication, Control and
  Computing}, October 2010, pp. 347--353.

\bibitem{Proietti02}
C.~Di, D.~Proietti, I.~E. Telatar, T.~J. Richardson, and R.~L. Urbanke,
  ``Finite-length analysis of low-density parity-check codes on the binary
  erasure channel,'' \emph{IEEE Trans. Inf. Theory}, vol.~48, no.~6, pp.
  1570--1579, Jun. 2002.

\bibitem{DAnantharam10}
L.~Dolecek, Z.~Zhang, V.~Anantharam, M.~Wainright, and B.~Nikolic, ``Analysis
  of absorbing sets and fully absorbing sets of array-based {LDPC} codes,''
  \emph{IEEE Trans. on Inf. Theory}, pp. 181--201, Jan. 2010.

\bibitem{YR18}
Y.~Yakimenka and E.~Rosnes, ``Failure analysis of the interval-passing
  algorithm for compressed sensing,'' in \emph{arXiv: 1806.05110v2. [Online].
  Available: https://arxiv.org/abs/1806.05110}, June 2018.

\bibitem{AREKD16}
B.~Amiri, A.~Reisizadehmobarakeh, H.~Esfahanizadeh, J.~Kliewer, and L.~Dolecek,
  ``Optimized design of finite-length separable circulant-based
  spatially-coupled codes: An absorbing set-based analysis,'' \emph{IEEE Trans.
  on Commun.}, vol.~64, no.~10, pp. 4029--4043, Oct 2016.

\bibitem{LiLin15}
J.~Li, S.~Lin, and K.~Abdel-Ghaffar, ``Improved message-passing algorithm for
  counting short cycles in bipartite graphs,'' in \emph{Proc. IEEE Int. Symp.
  Inf. Theory (ISIT)}, Jun. 2015, pp. 416--420.

\end{thebibliography}


\begin{thebibliography}{10}
\providecommand{\url}[1]{#1}
\csname url@samestyle\endcsname
\providecommand{\newblock}{\relax}
\providecommand{\bibinfo}[2]{#2}
\providecommand{\BIBentrySTDinterwordspacing}{\spaceskip=0pt\relax}
\providecommand{\BIBentryALTinterwordstretchfactor}{4}
\providecommand{\BIBentryALTinterwordspacing}{\spaceskip=\fontdimen2\font plus
\BIBentryALTinterwordstretchfactor\fontdimen3\font minus
  \fontdimen4\font\relax}
\providecommand{\BIBforeignlanguage}[2]{{%
\expandafter\ifx\csname l@#1\endcsname\relax
\typeout{** WARNING: IEEEtran.bst: No hyphenation pattern has been}%
\typeout{** loaded for the language `#1'. Using the pattern for}%
\typeout{** the default language instead.}%
\else
\language=\csname l@#1\endcsname
\fi
#2}}
\providecommand{\BIBdecl}{\relax}
\BIBdecl



\bibitem{AREKD16}
B.~Amiri, A.~Reisizadehmobarakeh, H.~Esfahanizadeh, J.~Kliewer, and L.~Dolecek,
  ``Optimized design of finite-length separable circulant-based
  spatially-coupled codes: An absorbing set-based analysis,'' \emph{IEEE
  Transactions on Communications}, vol.~64, no.~10, pp. 4029--4043, Oct 2016.

\bibitem{allerton}
A.~Beemer, S.~Habib, C.~Kelley, J.~Kliewer, ``A Generalized Algebraic Approach to Optimizing SC-LDPC Codes,'' submitted to Allerton Conference, 2017

@article{,
author={H. Esfahanizadeh and A. Hareedy and L. Dolecek}, 
title="A Novel Combinatorial Framework to Construct Spatially-Coupled Codes: Minimum Overlap Partitioning",
journal="",
ISSN="",
publisher="",
year="2017",
month="",
volume="",
number="",
pages="",
URL="",
DOI="",
}
\bibitem{HHD17}
H. Esfahanizadeh, A. Hareedy and L. Dolecek, ``A Generalized Algebraic Approach to Optimizing SC-LDPC Codes,'' IEEE International Symposium on Information Theory (ISIT), 2017
\bibitem{Linecount}
A file named linecount_m_1.pdf in Drop box

\end{thebibliography}

\end{document}